\documentclass[a4paper,UKenglish,cleveref, autoref,    thm-restate]{lipics-v2021}

 \usepackage{amsmath,amsfonts,amssymb,amsthm}
 \usepackage{mathbbol}
\usepackage{amsthm}
\usepackage{graphicx,color}
\usepackage{boxedminipage}
 \usepackage[linesnumbered, boxed, algosection,noline,noend]{algorithm2e}
 \usepackage{framed}
\usepackage{xspace}
\usepackage{todonotes}
\usepackage{xifthen}
 \usepackage{tabularx}
  \usepackage{capt-of}
\usepackage{thmtools}
 \usepackage{subcaption}
  \usepackage{booktabs}
\usepackage{calc}
\usepackage{dirtytalk}
\usepackage{bm}
\usepackage{enumitem}
 \usepackage{edtable}

\usepackage{mdframed}
\usepackage{multirow}
\usepackage{diagbox}

\usetikzlibrary{calc}
\usetikzlibrary{shapes}
\usetikzlibrary{arrows.meta}

\usepackage{mathrsfs}
 
\usepackage{array}
\newcolumntype{P}[1]{>{\centering\arraybackslash}p{#1}}



\newlength{\RoundedBoxWidth}
\newsavebox{\GrayRoundedBox}
\newenvironment{GrayBox}[1]%
   {\setlength{\RoundedBoxWidth}{.93\textwidth}
    \def\boxheading{#1}
    \begin{lrbox}{\GrayRoundedBox}
       \begin{minipage}{\RoundedBoxWidth}}%
   {   \end{minipage}
    \end{lrbox}
    \begin{center}
    \begin{tikzpicture}%
       \node(Text)[draw=black!20,fill=white,rounded corners,%
             inner sep=2ex,text width=\RoundedBoxWidth]%
             {\usebox{\GrayRoundedBox}};
        \coordinate(x) at (current bounding box.north west);
        \node [draw=white,rectangle,inner sep=3pt,anchor=north west,fill=white] 
        at ($(x)+(6pt,.75em)$) {\boxheading};
    \end{tikzpicture}
    \end{center}}     

\newenvironment{defproblemx}[2][]{\noindent\ignorespaces%
                                \FrameSep=6pt%
                                \parindent=0pt%
                \vspace*{-1.5em}
                \ifthenelse{\isempty{#1}}{%
                  \begin{GrayBox}{\textsc{#2}}%
                }{%
                  \begin{GrayBox}{\textsc{#2} parameterized by~{#1}}%
                }
                \begin{tabular*}{\textwidth}{@{\hspace{.1em}} >{\itshape} p{1.8cm} p{0.8\textwidth} @{}}%
            }{
                \end{tabular*}%
                \end{GrayBox}%
                \ignorespacesafterend
            }  

\title{Exact and Approximation Algorithms for Many-To-Many Point Matching in the Plane\footnote{A preliminary version of the paper has been accepted at ISAAC 2021.}} 

\author{Sayan Bandyapadhyay}{Department of Informatics, University of Bergen, Norway}{Sayan.Bandyapadhyay@uib.no}{https://orcid.org/0000-0001-8875-0102}{}
\author{Anil Maheshwari}{School of Computer Science, Carleton University}{anil@scs.carleton.ca}{orcid.org/0000-0002-1274-4598}{}
\author{Michiel Smid}{School of Computer Science, Carleton University}{michiel@scs.carleton.ca}{orcid.org/0000-0003-1955-4612}{}
\authorrunning{S. Bandyapadhyay, A. Maheshwari, M. Smid 
}
\Copyright{
Sayan Bandyapadhyay, Anil Maheshwari, Michiel Smid
}

\ccsdesc[100]{Theory of computation~Design and analysis of algorithms~Approximation algorithms analysis}

\keywords{Many-to-many matching, bipartite, planar, geometric, approximation} 

\category{} 

\supplement{}

\funding{This work is supported by the European Research Council (ERC) via grant LOPPRE, reference 819416 and by NSERC, Canada.}
\acknowledgements{We thank Saeed Mehrabi for introducing the many-to-many matching problem to us.}

\nolinenumbers 

\EventEditors{John Q. Open and Joan R. Access}
\EventNoEds{2}
\EventLongTitle{42nd Conference on Very Important Topics (CVIT 2016)}
\EventShortTitle{CVIT 2016}
\EventAcronym{CVIT}
\EventYear{2016}
\EventDate{December 24--27, 2016}
\EventLocation{Little Whinging, United Kingdom}
\EventLogo{}
\SeriesVolume{42}
\ArticleNo{23}

\begin{document}

\maketitle

\begin{abstract}
Given two sets $S$ and $T$ of points in the plane, of total size $n$, a {many-to-many} matching between $S$ and $T$ is a set of pairs $(p,q)$ such that $p\in S$, $q\in T$ and for each $r\in S\cup T$, $r$ appears in at least one such pair. The {cost of a pair} $(p,q)$ is the (Euclidean) distance between $p$ and $q$. In the {minimum-cost many-to-many matching} problem, the goal is to compute a many-to-many matching such that the sum of the costs of the pairs is minimized. This problem is a restricted version of minimum-weight edge cover in a bipartite graph, and hence can be solved in $O(n^3)$ time. In a more restricted setting where all the points are on a line, the problem can be solved in $O(n\log n)$ time [Colannino, Damian, Hurtado, Langerman, Meijer, Ramaswami, Souvaine, Toussaint; Graphs Comb., 2007]. However, no progress has been made in the general planar case in improving the cubic time bound. In this paper, we obtain an $O(n^2\cdot poly(\log n))$ time exact algorithm and an $O( n^{3/2}\cdot poly(\log n))$ time $(1+\epsilon)$-approximation in the planar case. Our results affirmatively address an open problem posed in [Colannino et al., Graphs Comb., 2007]. 
\end{abstract}

\section{Introduction}
Let $G=(V=S\cup T, E)$ be a simple bipartite graph where each edge has a non-negative real weight, and no vertex is isolated. The \emph{many-to-many} matching problem on $G$ is to find a subset of edges  $E'\subseteq E$ of minimum total weight such that for each vertex $v\in V$ there is an edge in $E'$ incident on $v$. This is often referred to as the  \emph{minimum-weight edge cover} problem. A standard method to compute a minimum-weight edge cover of $G$ in polynomial-time is to reduce the problem to the minimum-weight perfect matching problem on an equivalent graph, see \cite{schrijver-book,ColanninoDHLMRST07,eiter1997distance,FerdousPK18}. Since the reduction takes $O(|V|+|E|)$ time, the running time is the same as the fastest known algorithm for computing a minimum-weight perfect matching. A faster $\frac{3}{2}$-approximation algorithm is proposed in \cite{FerdousPK18}. 

Motivated by the computational problems in musical rhythm theory, Colannino et al. \cite{ColanninoDHLMRST07} studied the many-to-many matching problem in a geometric setting. Suppose we are given two sets $S$ and $T$ of points in the plane, of total size $n$. A \emph{many-to-many} matching between $S$ and $T$ is a set of pairs $(p,q)$ such that $p\in S$, $q\in T$ and for each $r\in S\cup T$, $r$ appears in at least one such pair. The \emph{cost of a pair} $(p,q)$ is the (Euclidean) distance $d(p,q)$ between $p$ and $q$. The \emph{cost of a set of pairs} is the sum of the costs of the pairs in the set. In the geometric \emph{minimum-cost many-to-many matching} problem, the goal is to compute a many-to-many matching such that the cost of the corresponding set of pairs is minimized. For points on a line,  an $O(n\log n)$ time dynamic-programming algorithm is proposed in \cite{ColanninoDHLMRST07}. In the same paper, the importance of the many-to-many matching problem for points in the plane is stated in the context of melody matching. Furthermore, the authors state the task of generalizing their work to this version in the plane as an important open problem, especially since geometry has helped in designing efficient algorithms for the computation of the minimum-weight perfect matching for points in the plane (see, e.g., \cite{vaidya1989geometry,varadarajan1999approximation}). Several  variants of the $1$-dimensional many-to-many matching problem are considered in \cite{Rajabi-AlniB16,abs-1904-05184,abs-1904-03015}.

\subsection{Our Results and Techniques}
In this work, we design several exact and approximation algorithms for minimum-cost many-to-many matching in the plane, thus affirmatively addressing the open problem posed in \cite{ColanninoDHLMRST07}.  First, we obtain an $O(n^{2}\cdot poly(\log n))$ time\footnote{We use the notation $poly()$ to denote a polynomial function} exact algorithm for this problem using a connection to minimum-weight perfect matching. We note that our time-bound matches the time bound for solving minimum-weight bipartite matching for points in the plane. Next, for any $\epsilon > 0$, we obtain a $(1+\epsilon)$-approximation algorithm for this problem with improved $O((1/\epsilon^c)\cdot n^{3/2}\cdot poly(\log n))$ running time for some small constant $c$. We also obtain a simple $2$-approximation in $O(n\log n)$ time. 

Next, we give an overview of our techniques. A major reason behind the scarcity of results for geometric many-to-many-matching is the lack of techniques to directly approach this problem. The $O(n^3)$ algorithm known for general graphs reduces the problem to (minimum-weight) bipartite perfect matching and uses a graph matching  algorithm to solve the problem on the new instance. However, this standard reduction \cite{keijsper1998efficient} from many-to-many matching to regular matching changes the weights of the edges in a convoluted manner. Hence, even if one starts with the planar Euclidean distances, the new interpoint distances in the constructed instance cannot be embedded in the plane or even in any metric space. As the algorithms for planar bipartite (perfect) matching heavily exploit the properties of the plane, they cannot be employed for solving the new instance of bipartite perfect matching. Thus, even though there is a wealth of literature for planar bipartite matching, no progress has been made in understanding the structure of many-to-many matching in the plane. 

In our approach, we use a rather unconventional connection to bipartite perfect matching. First, we use a different reduction to convert our instance of many-to-many matching in the plane to an equivalent instance of bipartite perfect matching. The new instance cannot be embedded in the plane, however it does not modify the original distances between the points. Rather the new bipartite graph constructed is the union of (i) the original geometric bipartite graph, (ii) a new bipartite clique all of whose edges have the  same weight and (iii) a linear number of additional edges. Thus, even though we could not use a planar matching algorithm directly, we could successfully use ideas from the literature of planar matching exploiting the structure of the constructed graph. A similar reduction was used in \cite{eiter1997distance} in the context of computation of link distance of graphs. However, we cannot use this reduction directly, as we cannot afford to explicitly store the constructed graph which contains $\Omega(n^2)$ edges. Recall that we are aiming for an $O((1/\epsilon^c)\cdot n^{3/2}\cdot poly(\log n))$ time bound. Nevertheless, we exploit the structure of this graph to implicitly store it in  $O(n)$ space. 

In Section \ref{sec:exact}, we show that the Hungarian algorithm can be implemented on our constructed graph in $O(n^{2}\cdot poly(\log n))$ time using ideas from planar matching. To obtain a subquadratic time bound, we implement the algorithm due to Gabow and Tarjan \cite{gabow1989faster}. A straightforward implementation of this algorithm might need $\Omega(n^2)$ time. We note that this algorithm has been used in several works on planar matching for obtaining efficient algorithms \cite{RaghvendraA20,varadarajan1999approximation,AsathullaKLR20}. Our algorithm is closest to the one in \cite{varadarajan1999approximation} among these algorithms from a moral point of view. The difference is that their algorithm is for planar points. However, we deal with a graph which is partly embeddable. Nevertheless, in Section \ref{thm:ptas-perfect}, we show that using additional ideas and data structures, the Gabow-Tarjan algorithm can be implemented in $O(n^{3/2}\cdot poly(\log n))$ time, albeit with a $(1+\epsilon)$-factor loss on the quality of the solution. 
In Section \ref{sec:2approx}, we describe our 2-approximation which the $(1+\epsilon)$-approximation algorithm uses as a subroutine for preprocessing.    

\section{Preliminaries}

In the \emph{bipartite perfect matching} problem, we are given an edge-weighted bipartite graph $G=(R,B,E)$ containing a perfect matching, and the goal is to find a perfect matching having the minimum cost or sum of the edge-weights. For our convenience, sometimes we would assume that the edges of $G$ are not given explicitly. For example, $R$ and $B$ might be two sets of points in the plane, and $G$ is the complete bipartite graph induced by the bipartition $(R,B)$. In this case, the points in $R\cup B$ can be used to implicitly represent the graph $G$. In our case, we will use similar implicit representation of input graphs for designing subquadratic algorithm. Now, we have the following lemma, which reduces our problem to an equivalent instance of bipartite perfect matching.

\begin{lemma}\label{lem:equiv-many-perfect}
Given an instance $I'$ of minimum-cost many-to-many matching, one can compute in $O(n\log n)$ time an instance of bipartite perfect matching $I$ such that (i) if there is a many-to-many matching for $I'$ of cost $C$, there is a perfect matching for $I$ of cost at most $C$, and (ii) if there is a perfect matching for $I$ of cost $C$, there is a many-to-many matching for $I'$ of cost $C$. 
\end{lemma}
\vspace{-6mm}

\begin{figure}[!ht]
		\begin{center}
			\includegraphics[width=.7\textwidth]{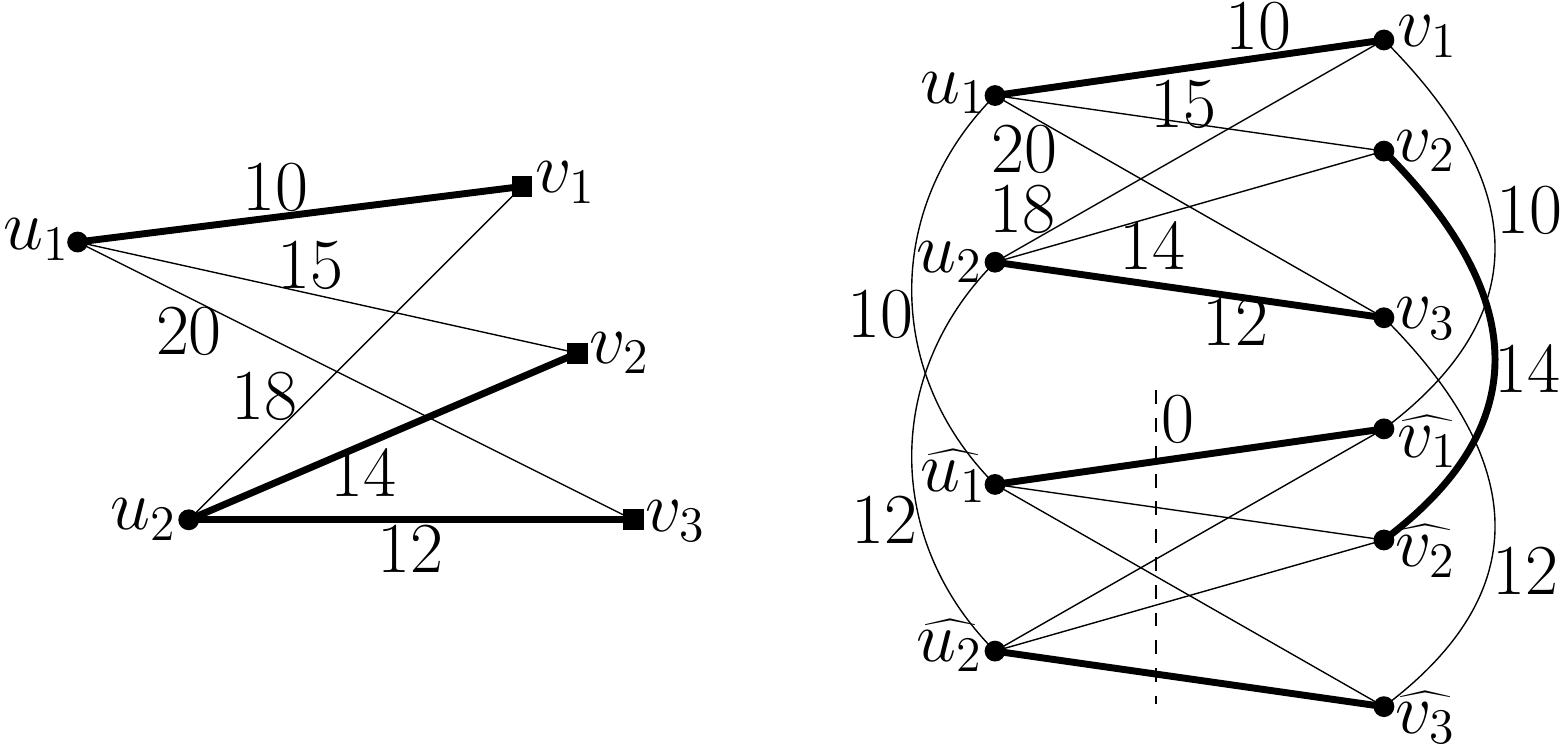}
			\vspace{-2mm}
			\caption{The figure on the left shows an instance $I'$ of minimum-cost many-to-many matching along with the interpoint distances, where $S=\{u_1,u_2\}$ and $T=\{v_1,v_2,v_3\}$. The right figure depicts the graph constructed from $I'$ along with the edge weights, where $R_0=\{u_1,u_2\}, R_1=\{\widehat{v_1},\widehat{v_2},\widehat{v_3}\}, B_0=\{v_1,v_2,v_3\}$ and $B_1=\{\widehat{u_1},\widehat{u_2}\}$. Solution pairs (or edges) are shown in bold.}
			\label{fig:graph}
		\end{center}
\end{figure}
\vspace{-10mm}
\begin{proof}
Given the instance $I'$ consisting of the two sets of points $S$ and $T$, we construct a bipartite graph $G=(R=R_0\cup R_1,B=B_0\cup B_1,E)$, where $R_0=S$, $B_0=T$, $R_1$ contains copies of the points in $T$, $B_1$ contains copies of the points in $S$. $E$ contains all the edges of $E_0=R_0\times B_0$ and $E_1=R_1\times B_1$, and also the ones in $E_2=\{(u,\hat{u})\mid u\in R_0,\hat{u}\text{ is the copy of } u \text{ in } B_1\}$ and $E_3=\{(\hat{v},v)\mid v\in B_0,\hat{v}\text{ is the copy of } v \text{ in } R_1\}$. The weight of each edge $(u,v)\in E_0$ is the distance between the points $u$ and $v$. The weight of each edge in $E_1$ is 0. The weight of any edge $(u,\hat{u})\in E_2$ is the distance between $u\in S$ and its closest neighbor in $T$. The weight of any edge $(\hat{v},v)\in E_3$ is the distance between $v\in T$ and its closest neighbor in $S$. See Figure \ref{fig:graph} for an example. Note that $G$ can be represented implicitly in $O(n)$ space, where $|S\cup T|=n$. Let $I$ be the constructed instance of bipartite perfect matching that consists of $G$. As the closest neighbors of $n$ points in the plane can be found in total $O(n\log n)$ time using Voronoi diagram (see Section  \ref{sec:2approx}), 
construction of $G$ takes $O(n\log n)$ time. 

Now, suppose $I'$ has a many-to-many matching $M'$. Consider the pairs in $M'$ as the edges of a graph $G'$ with vertices being the points in $S\cup T$. We will use the term pairs and edges in $G'$ interchangeably. First, note that wlog we can assume that $G'$ does not contain any path of length 3. Otherwise, we can remove the middle edge of such a path from $M'$, and $M'$ still remains a many-to-many matching. Thus, each component in $G'$ is a star. We compute a perfect matching in $G$ from $M'$ as follows. For each star in $G'$ having only one edge $(u,v)$ with $u\in S, v\in T$, add $(u,v)\in E_0$ to $M$. Also, add the edge $(\hat{v},\hat{u})\in E_1$ to $M$. Now, consider any star $H=\{(u,v_1),(u,v_2),\ldots,(u,v_t)\}$ in $G'$; wlog assume that $u\in S$. Add $(u,v_1)\in E_0$ to $M$. Also, add the edge $(\hat{v_1},\hat{u})\in E_1$ to $M$. For each $2\le i\le t$, add the edge $(\hat{v_i},v_i)\in E_3$ to $M$. It is not hard to verify that all the vertices of $G$ are matched in $M$. Also, as the weight of a star edge $(u,v_i)$ above is at least the weight of $(\hat{v_i},v_i)$ in $G$ by definition, the cost of $M$ is at most the cost of $M'$. 

Next, suppose $I$ has a perfect matching $M$; we construct a many-to-many matching $M'$ for $I'$. Consider any $u\in S=R_0$. If $(u,\hat{u})\in M$, add the pair $(u,u')$ to $M'$, where $u'$ is the closest neighbor of $u$ in $T$. Otherwise, $u$ is matched (in $M$) to some $v_1\in B_0$. In this case, simply add the pair $(u,v_1)$ to $M'$. Similarly, consider any $v\in T=B_0$. If $(\hat{v},v)\in M$, add the pair $(v',v)$ to $M'$, where $v'$ is the closest neighbor of $v$ in $S$. Otherwise, $v$ is matched (in $M$) to some $u_1\in R_0$. In this case, simply add the pair $(u_1,v)$ to $M'$. It is not hard to verify that $M'$ is a many-to-many matching, and the cost of $M'$ is same as the cost of $M$. 
\end{proof}

Now, consider any matching in a graph. An \emph{alternating} path is a path whose edges alternate between matched and unmatched edges. Similarly, one can define \emph{alternating} cycles and trees. A vertex is called \emph{free} if it is not matched. An \emph{augmenting} path is an alternating path which starts and ends at free vertices. Given an augmenting path $P$ w.r.t. a matching $M$, we can augment $M$ by one edge if we remove the edges of $P\cap M$ from $M$ and add the edges in $P\setminus M$ to $M$. The new matching is denoted by $M \oplus P$. Throughout the paper, $m$ and $n$ denote the number of edges and vertices, respectively, unless otherwise specified. We denote the weight or cost of an edge $(u,v)$ by $c(u,v)$. In our discussions, a path can be treated as an ordered set of vertices or edges depending on the context.


\section{An Exact Algorithm}
\label{sec:exact}

Consider the instance $I$ obtained by the reduction in Lemma \ref{lem:equiv-many-perfect}. In this section, we prove the following theorem. 

\begin{theorem}\label{thm:exact-perfect}
Bipartite perfect matching can be solved exactly on $I$ in time $O(n^{2}\cdot poly(\log n))$, and hence there is an $O(n^{2}\cdot poly(\log n))$ time exact algorithm for minimum-cost many-to-many matching. 
\end{theorem}


In the rest of this section, we prove Theorem \ref{thm:exact-perfect}. To prove this theorem we use the Hungarian algorithm \cite{kuhn1956variants} for computing a (minimum-cost) bipartite perfect matching. 

In the Hungarian algorithm, there is a dual variable $y(v)$ corresponding to each vertex $v$. Every feasible matching $M$ must satisfy the following two conditions. 
\begin{align}
    y(u)+y(v) & \le c(u,v) \text{ for every edge } (u,v)\label{eq:1}\\
    y(u)+y(v) &= c(u,v) \text{ for every edge } (u,v) \in M\label{eq:2}
\end{align}
Given any matching $M$ and an augmenting path $P$, the net-cost or augmentation cost is defined as, 
\begin{align*}
    \phi(P) = \sum_{(u,v) \in P\setminus M} c(u,v)-\sum_{(u,v) \in P\cap M} c(u,v)
\end{align*}
$\phi(P)$ is basically the cost increment for augmenting $M$ along $P$. The net-cost of any alternating cycle can be defined in the same way. The Hungarian algorithm starts with an empty matching $M$. In every iteration, it computes an augmenting path of the minimum net-cost and augments the current matching. The algorithm halts once a perfect matching is found. If there is a perfect matching in the graph, this algorithm returns one after at most $n$ iterations. Moreover, an augmenting path can be found in $O(m)$ time leading to $O(mn)$ running time in total. 

It is possible to show that any perfect matching is a minimum-cost matching if and only if there is no negative net-cost alternating cycle with respect to it. Moreover, a feasible matching with dual values $\{y(v)\}$ satisfies this property. Thus, it is sufficient to find a perfect matching that is feasible. 

For finding an augmenting path, a Hungarian search procedure is employed. Hungarian search uses Dijkstra's shortest path algorithm to find a minimum net-cost path in the graph where the value $y(u)+y(v)$ is subtracted from the weight of each edge $(u,v)$. This along with the first feasibility condition ensure that each edge has a non-negative weight, and hence there is no negative cycle in the graph. So, one can correctly employ Dijkstra's algorithm to find such a shortest path. Finally, one can show that the minimum net-cost augmenting path with original weights corresponds to a shortest path with the modified weights, and vice versa. After augmenting the current matching with the newly found path, the dual values are adjusted appropriately to ensure feasibility of the new matching. 

Now, we describe in detail how the Hungarian search procedure is implemented in each iteration. An edge $(u,v)$ is called \emph{admissible} if $y(u)+y(v)=c(u,v)$. It can be shown that it suffices to find an augmenting path consisting of only admissible  edges. Let $M$ be the current matching and $F$ be the free vertices of $R$. To obtain the desired augmenting path, a forest $\mathcal{F}$ is grown whose roots are in $F$. Each tree in $\mathcal{F}$ is an augmenting tree rooted at a vertex in $F$. Once $\mathcal{F}$ contains an augmenting path the search is completed. At any moment, let $R'$ be the vertices of $R$ in $\mathcal{F}$ and $B'$ be the vertices of $B$ not in $\mathcal{F}$. Initially, $R'=F$ and $B'=B$. Also, let \[\delta=\min_{u\in R',v\in B'} \{c(u,v)-y(u)-y(v)\}.\] Note that $\delta=0$ means there is an admissible edge. In each step, if $\delta=0$, an admissible edge $(u,v)$ is selected where $u\in R'$ and $v\in B'$. If $v$ is free, $(u,v)$ is added to $\mathcal{F}$ and the desired augmenting path is found. Otherwise, let $v$ be matched to $u'$ (which is not in $\mathcal{F}$ by an invariant). In this case, the edges $(u,v)$ and $(v,u')$ are added to $\mathcal{F}$; $u'$ is added to $R'$ and $v$ is removed from $B'$. 

If at some moment, $\delta$ becomes more than 0 (no admissible edge), we perform dual adjustments. In particular, for each $u\in R'$, $y(u)$ is updated to $y(u)+\delta$ and for each $v\in \mathcal{F}\cap B$, $y(v)$ is updated to $y(v)-\delta$. This ensures that at least one edge becomes admissible, e.g., the edge corresponding to which the $\delta$ value is achieved. Thus, eventually the search halts with an augmenting path in $\mathcal{F}$. 

It can be shown  that if $\delta$ can be computed efficiently, then the desired augmenting path can also be found efficiently \cite{gabow1989faster,AgarwalES99,vaidya1989geometry}. For that purpose, another variable $\Delta$ is maintained. Also, for each vertex $v$, a weight $\sigma_v$ is stored. In the beginning of each step, $\Delta=0$, $\sigma_v=y(v)$. When the edges $(u,v)$ and $(v,u')$ are added to $\mathcal{F}$, $\sigma_{u'}$ is updated to $y(u')-\Delta$ and $\sigma_{v}$ is updated to $y(v)+\Delta$. Once $\delta$ becomes more than 0, $\Delta$ is updated to $\Delta+\delta$. 

Note that the weight of a vertex is updated only once when it is added to $\mathcal{F}$. We have the following observation. 

\begin{observation}
\cite{gabow1989faster,AgarwalES99,vaidya1989geometry} For each $u\in R'$, the current dual value of $u$, $y(u)$, is equal to $\sigma_u+\Delta$. For each $v\in \mathcal{F}\cap B$, the current dual value of $v$, $y(v)$, is equal to $\sigma_v-\Delta$. 
\end{observation}

It follows from the above observation that $\delta$ can be equivalently expressed as follows. \[\delta=\min_{u\in R',v\in B'} \{c(u,v)-\sigma_u-\sigma_v\}-\Delta. \]

Hence, Hungarian search boils down to the following task ignoring the trivial details. We need to maintain two sets $R'\subseteq R$ and $B'\subseteq B$. Initially, $B'=B$. In each step, a vertex $r$ is added to $R'$ and a vertex $b$ is removed from $B'$. Additionally, each vertex $v$ has a weight $\sigma_v$. In every step, the goal is to maintain the bichromatic closest pair, which is the pair $(r,b)\in R'\times B'$ with the minimum $c(r,b)-\sigma_r-\sigma_b$ value. In the following, we construct a data structure that can be used to perform the above task (a Hungarian search) in $O(n\cdot poly(\log n))$ time. As we need at most $n$ such searches, Theorem \ref{thm:exact-perfect} follows. 

Recall that in our instance $I$, we are given the graph $G=(R=R_0\cup R_1,B=B_0\cup B_1,E)$, where $R_1$ contains copies of the vertices in $B_0$, $B_1$ contains copies of the vertices in $R_0$. $E=E_0\cup E_1\cup E_2\cup E_3$, where $E_0=R_0\times B_0$, $E_1=R_1\times B_1$, $E_2=\{(u,\hat{u})\mid u\in R_0,\hat{u}\in B_1\}$, and $E_3=\{(\hat{v},v)\mid v\in B_0,\hat{v}\in R_1\}$. Let $n=|R|=|B|$. 

Our data structure $\mathcal{D}$ is a collection of three data structures. First, we construct the dynamic bichromatic closest pair data structure from \cite{kaplan2020dynamic} for the two point sets $R_0$ and $B_0$ with the distance function $c(r,b)-\sigma_r-\sigma_b$ for each pair $(r,b)$. We refer to this data structure as $\mathcal{D}_1$. Initially, it contains only the points in $R'\cup B'$. Next, we construct two max-heaps $H_1^r$ and $H_1^b$ for the vertices in $R_1\cap R'$ and $B_1\cap B'$, respectively. Initially, $H_1^r$ contains vertices in $R'$ and $H_1^b$ contains all the vertices in $B_1$. The key value of each vertex $v$ is its weight $\sigma_v$. Using $H_1^r$ and $H_1^b$, the pair $(u,v)\in (R_1\cap R')\times (B_1\cap B')$ with the maximum $\sigma_u+\sigma_v$ value can be found in $O(1)$ time. 
We also construct another min-heap $H_{23}$ to store the edges in $E_2\cup E_3$ with key value $c(r,b)-\sigma_r-\sigma_b$ for each pair $(r,b)$. Initially, it is empty. In every iteration, it contains only those edges $(u,v)$ such that $u\in R'$ and $v\in B'$. We also maintain a global closest pair $(r,b)$ over the three data structures $\mathcal{D}_1$, $H_1^r\cup H_1^b$ and $H_{23}$ with the minimum key value $c(r,b)-\sigma_r-\sigma_b$. 

Using $\mathcal{D}$, we implement each step as follows. If $v\in B_0$, remove $v$ from $\mathcal{D}_1$. Also remove $(\hat{v},v)$ from $H_{23}$ if it is in $H_{23}$. If $v\in B_1$, remove $v$ from $H_1^b$. Also remove the edge $(u',v)$ with $u'\in R_0$ from $H_{23}$ if its in $H_{23}$. If $u\in R_0$, add $u$ to $\mathcal{D}_1$. Also add the edge $(u,\hat{u})$ to $H_{23}$ if $\hat{u}\in B'$. If $u\in R_1$, add $u$ to $H_1^r$. Also add the edge $(u,v')$ with $v'\in B_0$ to $H_{23}$ if $v'\in B'$. 

It is not hard to verify that at each step the correct global closest pair is stored in $\mathcal{D}$. The correctness for  $\mathcal{D}_1$ and $H_{23}$ follow trivially. Also, as the weight of the edges in $E_1$ are same, it is sufficient to find a pair $(u,v)\in E_1$ for which $\sigma_u+\sigma_v$ value is maximized. As $E_1$ contains all the edges in $R_1\times B_1$, equivalently it suffices to find a $u\in R_1$ with the maximum $\sigma_u$ value and a $v\in B_1$ with the maximum $\sigma_v$ value. 

We note that the total number of steps in a Hungarian search is at most $n$. Also, $O(n)$ operations (insertions, deletions and searching) on $\mathcal{D}_1$ can be performed in amortized $O(n\cdot poly(\log n))$ time \cite{kaplan2020dynamic}. The construction time and space of $\mathcal{D}_1$ are also $O(n\cdot poly(\log n))$. Moreover, the $O(n)$ operations on all max-heaps can be performed in total $O(n\log n)$ time. It follows that Hungarian search can be implemented in $O(n\cdot poly(\log n))$ time leading to the desired running time of our algorithm.

\section{A 2-approximation in $O(n\log n)$ Time}
\label{sec:2approx}

In this section, we prove the following theorem.

\begin{theorem}\label{thm:2approx}
There is an $O(n\log n)$ time 2-approximation algorithm for minimum-cost many-to-many matching.  
\end{theorem}

Our algorithm is as follows. For each point $u\in S$, we add the pair $(u,q)$ to the solution matching $M$ which minimizes $d(u,q)$ over all $q\in T$. Similarly, for each point $v\in T$, we add the pair $(p,v)$ to $M$ which minimizes $d(p,v)$ over all $p\in S$. Thus, for each point in a set, we basically add the pair corresponding to its nearest neighbor in the other set. By definition, the computed solution $M$ is a many-to-many matching. Next, we argue how to implement this algorithm in $O(n\log n)$ time. First, we compute the Voronoi diagram of the points in $T$, compute a triangulation of the Voronoi cells and then construct Kirkpatrick's planar point location data structure \cite{Kirkpatrick83} using the triangulation. Using this data structure, for any $u\in S$, its nearest neighbor in $T$ can be computed in $O(\log n)$ time. Also, the data structure uses $O(n)$ space and $O(n\log n)$ construction time including the preprocessing time. Using a similar data structure for points in $S$, we can compute the nearest neighbors of the points in $T$ in $O(n\log n)$ time. The following lemma completes the proof of Theorem \ref{thm:2approx}.     

\begin{lemma}
The cost of the matching $M$ computed as in above is at most twice the optimal cost. 
\end{lemma}

\begin{proof}
For each point $p$ in $S$ (resp. in $T$), let $\eta(p)$ be its nearest neighbor in $T$ (resp. in $S$). Note that the cost of $M$, cost$(M)$ $=\sum_{p\in S\cup T} d(p,\eta(p))$. Now, consider any optimal matching $M^*$. For each $p \in S\cup T$, let $c^*(p)$ be the cost of any pair $(p,q^*)$ in $M^*$, i.e., a pair containing $p$. Then, $c^*(p)$ must be at least the distance between $p$ and its nearest neighbor in the other set, i.e.,  $c^*(p)\ge d(p,\eta(p))$. It follows that, the cost of $M$,  \[\sum_{p\in S\cup T} d(p,\eta(p))\le \sum_{p\in S\cup T} c^*(p)=\sum_{p_1\in S} c^*(p_1)+\sum_{p_2\in T} c^*(p_2)\le \text{cost}(M^*)+\text{cost}(M^*)=2\cdot \text{cost}(M^*).\] 
\end{proof}

\section{An Improved $(1+\epsilon)$-approximation}
Consider the instance $I$ obtained by the reduction in Lemma \ref{lem:equiv-many-perfect}. In this section, we prove the following theorem. 

\begin{theorem}\label{thm:ptas-perfect}
A $(1+\epsilon)$-approximate bipartite perfect matching for $I$ can be computed in time $O((1/\epsilon^c)\cdot n^{3/2}\cdot poly(\log n))$ for some constant $c$, and hence there is an $O((1/\epsilon^c)\cdot n^{3/2}\cdot poly(\log n))$ time $(1+\epsilon)$-approximation algorithm for minimum-cost many-to-many matching. 
\end{theorem}



In the rest of this section, we prove Theorem \ref{thm:ptas-perfect}. In particular, we show that the bipartite perfect matching algorithm by Gabow and Tarjan \cite{gabow1989faster} can be implemented on the instance $I$ in the mentioned time. 

The Gabow-Tarjan algorithm is based on a popular scheme called the bit-scaling paradigm. The algorithm is motivated by two classic matching algorithms: Hopcroft-Karp \cite{hopcroft1973n} for maximum cardinality bipartite matching with $O(m\sqrt{n})$ running time and Hungarian algorithm \cite{kuhn1956variants} for minimum-cost bipartite matching with $O(mn)$ running time. The Hopcroft-Karp algorithm chooses an augmenting path of the shortest length. The Hungarian algorithm, as mentioned before, chooses an augmenting path whose augmentation cost is the minimum. When the weights on the edges are small an augmenting path of the shortest length approximates the latter path. The Gabow-Tarjan algorithm scales the weights in a manner so that all the effective weights are small. This helps to combine the ideas of the two algorithms, which leads towards an $O(m\sqrt{n}\log (n N))$ time algorithm for (minimum-cost) bipartite perfect matching, where $N$ is the largest edge weight. 

Next, we describe the Gabow-Tarjan algorithm. This algorithm is based on the ideas of the Hungarian algorithm. However, here instead of a feasible matching we compute a \emph{1-feasible} matching. A matching $M$ is called 1-feasible if it satisfies the following two conditions.    
\begin{align}
    y(u)+y(v) & \le c(u,v)+1 \text{ for every edge } (u,v)\label{eq:3}\\
    y(u)+y(v) &= c(u,v) \text{ for every edge } (u,v) \in M\label{eq:4}
\end{align}
Note that the only difference is that now the sum of dual variables $y(u)+y(v)$ can be 1 plus the cost of the edge. This additive error of 1 on every unmatched edge ensures that longer augmenting paths have larger cost. As an effect, the algorithm picks short augmenting paths as in the Hopcroft-Karp algorithm. 

A \emph{1-optimal} matching is a perfect 1-feasible matching. Note that a 1-optimal matching costs more than the original optimal matching. However, as the error is at most +1 for every edge, one can show the following.

\begin{lemma}
\cite{gabow1989faster} Let $M$ be a 1-optimal matching and $M'$ be any perfect matching. Then $c(M')\ge c(M)-n$. 
\end{lemma}

It follows from the above lemma that, to annihilate the error introduced, one can scale the weight of each edge by a factor of $(n+1)$, and then with the scaled weights the cost of a 1-optimal matching is same as the cost of any optimal matching. Let $\overline{c}(u,v)$ be the scaled cost of $(u,v)$, i.e., $\overline{c}(u,v)=(n+1)\cdot {c}(u,v)$. Let $k=\lfloor\log ((n+1) N)\rfloor+1$ be the maximum number of bits needed to represent any new weight. 

The Gabow-Tarjan algorithm runs in $k$ different scales. In each scale $i$ ($1\le i\le k$), the most significant $i$ bits of $\overline{c}(u,v)$ are used for defining the current cost of each edge $(u,v)$. The dual values are also modified to maintain 1-feasibility of an already computed perfect matching in the following way: $y(v)\leftarrow 2y(v)-1$ for every vertex $v$. Then with the current edge costs and dual values, the algorithm computes a 1-optimal matching. 

By the above claim, that any 1-optimal matching is also optimal with scaling factor $n+1$, the 1-optimal matching computed by the algorithm at $k$-th scale must be optimal. 

To find a 1-optimal matching on a particular scale a procedure called \emph{match} is employed which we describe below. Before that we need a definition. Consider any 1-feasible matching $M$. An edge $(u,v)$ is called \emph{eligible} if it is in $M$ or $y(u)+y(v)=c(u,v)+1$. It can be shown that for the purpose of computing a 1-optimal matching it suffices to consider the augmenting paths which consist of eligible edges only. 

\vspace{2mm}
\noindent\fbox{
\parbox{\textwidth}{
\noindent \textbf{The} \emph{match} \textbf{procedure}

\noindent Initialize all the dual variables $y(v)$ to 0 and $M$ to $\emptyset$. Repeat the following two steps until a perfect matching is obtained in step 1. 
\begin{enumerate}
    \item Find a maximal set $\mathcal{A}$ of augmenting paths of eligible edges. For each path $P\in \mathcal{A}$, augment the current matching $M$ along $P$ to obtain a new matching which is also denoted by $M$. For each vertex $v\in P\cap B$, decrease $y(v)$ by 1. (This is to ensure that the new matching $M$ also is 1-feasible.) If $M$ is perfect, terminate. 
    \item Employ a Hungarian search to adjust the values of the dual variables (by keeping $M$ 1-feasible), and find an augmenting path of eligible edges. 
\end{enumerate}
}}
\vspace{2mm}

Note that the number of free vertices in every iteration of \emph{match} is at least 1 less than that in the previous iteration, as step 2 always ends with finding an augmenting path. By also showing that the dual value of a variable is increased by at least 1 in every call of Hungarian search, they proved that $O(\sqrt{n})$ iterations are sufficient to obtain a 1-optimal matching. It can be shown that each iteration of \emph{match} can be executed in general bipartite graphs in $O(m)$ time leading to the complexity of $O(m\sqrt{n})$ in each scale. As we have $O(\log (nN))$ scales, in total the running time is $O(m\sqrt{n}\log (nN))$.  

Next, we use the Gabow-Tarjan algorithm to compute a perfect matching for our instance of bipartite perfect matching. In particular, we show that by exploiting the structure of our instance, it is possible to implement every iteration of \emph{match} in $O(n \cdot poly(\log n))$ time. 
First, we show that one can consider a modified instance with bounded aspect ratio of the weights, for the purpose of computing a $(1+\epsilon)$-approximation. 



Let OPT be the optimal cost of perfect matching on the instance $I$. Recall that the instance $I$ consists of the graph $G=(R=R_0\cup R_1,B=B_0\cup B_1,E)$. Given the implicit representation of $G$, we compute a 2-approximate solution for minimum-cost many-to-many matching on the two sets of points $R_0$ and $B_0$ using the algorithm in Theorem \ref{thm:2approx}. Let $C$ be the cost of this solution. By Lemma \ref{lem:equiv-many-perfect}, $\text{OPT}\le C\le 2\text{OPT}$. We construct a new instance $I_1$ which consists of the implicit representation of the same graph $G$. Additionally, we assume that any edge with weight more than $C$ in $I_1$ will not be part of the solution, and each edge has cost at least  $\epsilon C/(2n)$. Note that it is not possible to explicitly set the weight of each and every edge if we are allowed to spend $o(n^2)$ time. Thus, for the time being, we make the above assumptions implicitly. Later, we will make them explicit in our algorithm. Note that the construction time of $I_1$ is dominated by the time of the 2-approximation algorithm, which is $O(n\log n)$. We obtain the following lemma.     

\begin{restatable}{lemma}{boundedaspectratio}\label{lem:bounded-aspect-ratio}
Given $I$, one can compute in $O(n\log n)$ time another instance $I_1$ of bipartite perfect matching, such that (i) the weight of every edge is in $[\epsilon C/(2n),C]$ where $\text{\emph{OPT}}\le C\le 2\text{\emph{OPT}}$, (ii) $I_1$ has a perfect matching of weight at most $(1+\epsilon) \text{\emph{OPT}}$, and (iii) any perfect matching in $I_1$ of weight $C'$ is also a perfect matching in $I$ of weight at most $C'$.         
\end{restatable}

\begin{proof}
(i) follows by the above construction of $I_1$. (ii) follows from the facts that OPT does not contain any edge of weight more than $C$ and the weight of an edge in OPT is increased by at most $\epsilon C/(2n)\le \epsilon \text{OPT}/n$ in $I_1$. (iii) follows, as each edge in $I_1$ is also present in $I$ with possibly equal or lesser cost.    
\end{proof}

Henceforth, we solve the problem on $I_1$. Note that the minimum edge weight in $I_1$ is $\epsilon C/(2n)$ and the maximum is $C$. By scaling the weights by $2n/(\epsilon C)$, we can assume wlog that the edge weights in $I_1$ are in $[1,n^2]$. Moreover, we can assume that each edge weight $w$ is rounded up to the nearest integer at least $w$. We can afford to remove the fractions, as each fraction costs less than 1, which is at most $\epsilon \text{OPT}/n$ w.r.t. the original weights. For our convenience, we also divide the weights into $O(\log_{1+\epsilon} n)$ classes as follows. For each weight $a$ (an integer) with $(1+\epsilon)^i \le a < (1+\epsilon)^{i+1}$, $a$ is rounded to the largest integer in the range $[(1+\epsilon)^i,(1+\epsilon)^{i+1})$, where $0\le i\le \lceil 2\log_{1+\epsilon} n\rceil$. We denote this largest integer corresponding to the $i$-th weight class by $w_i$. We note that the above weight scalings are performed implicitly. It is not hard to verify that these still preserve a $(1+O(\epsilon))$-approximate solution, which is sufficient for our purpose. Henceforth, we treat $\epsilon$ as a constant and hide function of $\epsilon$ in time complexity as a constant in $O()$ notation. It will not be hard to verify that the dependency on $\epsilon$ that we hide is $(1/\epsilon)^c$ for some true small constant $c$.   

To implement the Gabow-Tarjan algorithm, we show how the \emph{match} procedure can be implemented efficiently. To implement step 1 of \emph{match}, we store the information about the input graph $G$ in a data structure that we refer to as  MATCH. We allow the following operation on MATCH. In the following, we denote the current matching by $M$.

\vspace{-4mm}
\subparagraph*{FIND$\_$MAXIMAL$\_$APS.} Find a maximal set of vertex disjoint augmenting paths of eligible edges with respect to $M$. 



Given the MATCH data structure, we implement the step 1 of \emph{match} as follows. The dual values are stored in an array indexed by the vertices. Note that the dual values remain fixed in step 1 while the augmenting paths are found. Afterwards, the dual values are updated.  We first make a call to FIND$\_$MAXIMAL$\_$APS to obtain a maximal set $\mathcal{A}$ of paths. For each $P\in \mathcal{A}$, we augment $M$ along $P$ to obtain a new matching $M$. Also, for each $v\in P\cap B$, we decrease $y(v)$ by 1.  

In Section \ref{sec:DS}, we show how to construct and maintain {MATCH} so that the above subroutine can be performed in time $O(n \cdot poly(\log n))$. The building time of {MATCH} is $O(n \cdot poly(\log n))$, and it takes $O(n \cdot poly(\log n))$ space. Thus, by noting that $\mathcal{A}$ contains disjoint paths, step 1 can be implemented in $O(n \cdot poly(\log n))$ time and space. 

We also need another data structure which will help us implement step 2 of \emph{match}, which we refer to as the Hungarian search data structure. As described before, 
Hungarian search boils down to 
maintaining a bichromatic closest pair $(r,b)$ of two sets 
with the minimum $c(r,b)-\sigma_r-\sigma_b$ value. 
However, here we have to be more careful, as for an unmatched eligible edge $(u,v)$, $y(u)+y(v)=c(u,v)+1$. In contrast, in the Hungarian algorithm, we had $y(u)+y(v)=c(u,v)$ in that case. Hence, we have to consider the $c(u,v)+1-y(u)-y(v)$ value as the distance of such an unmatched pair $(u,v)$. This apparently makes our life harder, as now we have to deal with two types of distance functions: one for matched pairs and one for unmatched pairs. However, we use the following observation to consider only one type of distances, which follows from our description of Hungarian search in Section \ref{sec:exact}.    

\begin{observation}
Consider any matched edge $(u,v)$ with $u\in R$ and $v\in B$. In the Hungarian search, if $u\in R'$, i.e., $u$ is already in the forest, then $v$ must also be in the forest, i.e., $v\notin B'$.  
\end{observation}

Recall that for computing $\delta$, we look into the pairs $(u,v)$ where $u\in R'$ and $v\in B'$. By the above observation, it suffices to probe only unmatched edges. Hence, we can again work with only one distance function $c(r,b)+1-\sigma_r-\sigma_b$ for the purpose of computing $\delta$. As the $+1$ term is common in all distances, we also drop that, and work with our old distance function.  

In Section \ref{sec:DS}, we show how to construct and maintain Hungarian search data structure so that the task of maintaining closest pair can be performed  in total $O(n \cdot poly(\log n))$ time. Moreover, the data structure uses $O(n \cdot poly(\log n))$ construction time and space. 

\begin{lemma}
Using the MATCH and Hungarian search data structures, one can implement the Gabow-Tarjan algorithm on the instance $I_1$ in time $O(n^{3/2}\cdot poly(\log n))$, i.e., a minimum cost bipartite perfect matching in $I_1$ can be computed in time $O(n^{3/2}\cdot poly(\log n))$.  
\end{lemma}


\section{Data Structures}
\label{sec:DS}
\subsection{The MATCH Data Structure}
We would like to construct a data structure where given a matching at a fixed scale, a maximal set of augmenting paths can be computed efficiently. In particular, we are given the graph $G=(R=R_0\cup R_1,B=B_0\cup B_1,E)$, where $R_1$ contains copies of the vertices in $B_0$, $B_1$ contains copies of the vertices in $R_0$. $E=E_0\cup E_1\cup E_2\cup E_3$, where $E_0=R_0\times B_0$, $E_1=R_1\times B_1$, $E_2=\{(u,\hat{u})\mid u\in R_0,\hat{u}\in B_1\}$, and $E_3=\{(\hat{v},v)\mid v\in B_0,\hat{v}\in R_1\}$. Let $n=|R|=|B|$. 
Also, note that each edge weight is an integer $w_i$ for $0\le i\le \lceil (c'\log n)/\epsilon\rceil$, where $c'$ is a constant. Let $E_i^0$ be the set of edges in $R_0\times B_0$ with weights $w_i$. We define a bi-clique cover for each $E_i^0$ as a collection $\mathcal{C}_i=\{(P_{i1},Q_{i1}),(P_{i2},Q_{i2}),\ldots,(P_{it(i)},Q_{it(i)})\}$ where $P_{ij}\subseteq R_0$, $Q_{ij}\subseteq B_0$, all the edges in $P_{ij}\times Q_{ij}$ are in $E_i^0$ and $\cup_{j=1}^{t(i)} P_{ij}\times Q_{ij} = E_i^0$. The size of $E_i^0$ is  $\sum_{j=1}^{t(i)} |P_{ij}|+|Q_{ij}|$. Given the points in $R_0$ and $B_0$, using standard range searching data structures, one can compute such a bi-clique cover of size $O((n/{\epsilon})\log^2 n)$ in $O((n/{\epsilon})\log^2 n)$ time. We note that bi-clique covers are also used for the algorithm in \cite{varadarajan1999approximation}. Let $\mathcal{C}=\mathcal{C}_i$. Thus $\mathcal{C}$ can be computed in  $O((n/{\epsilon}^2)\log^3 n)$ time and space. 

In MATCH we store the bi-clique covers in $\mathcal{C}$ corresponding to the edges in $E_0$. Also, we store the edges in $E_2$ along with their weights in an array $A_2$ indexed by the vertices of $R_0$ and the edges in $E_3$ and their weights in an array $A_3$ indexed by the vertices of $B_0$. For finding an augmenting path efficiently, we need to store additional information. Before describing that, we describe in more detail how the augmenting paths are found.   

The maximal set of augmenting paths are found by a careful implementation of depth first search. In this implementation, vertices can be labeled as marked. Initially all vertices are unmarked. We select any free unmarked vertex of $R$ and initialize a path $P$ at that vertex. $P$ is extended from the last vertex $u$ (in $R$ as an invariant) as follows. We probe an eligible edge $(u,v)$. If $v$ is already marked, the next eligible edge is considered. If no such edge exists, the last two edges (one unmatched and one matched) are deleted from $P$. If $P$ becomes empty, a new path is initialized. For the remaining cases, the following subroutine is called. 
\vspace{-2mm}

\subparagraph*{AUGMENTING$\_$PATH$(v)$.} If $v$ is unmarked and free, we have found an augmenting path; $v$ is marked, $P$ is added to $\mathcal{A}$ and a new path is initialized. In this case, return DONE. If $v$ is unmarked, but matched with another vertex $w$, $(u,v)$ and $(v,w)$ are added to $P$; $v,w$ are marked and the extension continues from $w$ (in $R$).   

\begin{restatable}{observation}{unmatchedonly}\label{obs:unmatched_only}
In the above procedure, we always maintain the invariant that when we extend a path it suffices to probe only unmatched edges. 
\end{restatable}
\begin{proof}
\begin{claim}
In AUGMENTING$\_$PATH$(v)$, we always maintain the invariants that when we extend a path $P$ from the last vertex $u$, (i) $u\in R$ and (ii) if $u$ is not the first vertex of $P$, the last edge $(v,u)\in P$ is a matched edge.  
\end{claim}

The above claim helps us restrict the search from a vertex of $R$. Now, whenever we extend the path $P$ from the last vertex $u$, there are two cases. Either $P$ contains a single vertex $u$, and in this case $u$ is free; we probe an unmatched eligible edge $(u,v)$. Or, $P$ contains at least two edges, and in this case the last edge of $P$ is a matched edge by the above observation; we need to probe an unmatched eligible edge $(u,v)$.  Thus, it suffices to probe only unmatched edges for the purpose of extension of $P$. 
\end{proof}

Note that in the above procedure we cannot afford to probe all the unmatched  edges. So, we have to implement the above step carefully. First, note that an edge $(u,v)$ ($u \in R$) is never scanned twice. When $(u,v)$ is probed the first time, if $v$ is unmarked, it becomes marked in all the cases. Also, once $v$ is marked, $(u,v)$ is never used to extend $P$. Thus, we can eliminate $(u,v)$ from further probing. 

From the above discussion, as $E_2$ and $E_3$ contain $O(n)$ edges in total they can be probed in $O(n)$ time. However, $E_0$ and $E_1$ contain $\Omega(n^2)$ edges and thus for probing them we need specialized data structures. Next, we describe those. 

Let $c(u,v)$ be the weight of $(u,v)$ at the current scale. $y(u)$ is the dual value of the vertex $u$ which remains fixed throughout the augmenting paths finding process. Again consider the bi-clique covers in $\mathcal{C}$. For each $0\le i\le \lceil (c'\log n)/\epsilon\rceil$, $w_i$ denotes the weight of the edges in $i$-th class. Also, let $\ell$ be the weight class to which the edges in $E_1$ belong, i.e., all of their weights are $w_{\ell}$. For each such $i$ and $1\le j\le t(i)$, we store in \emph{match} the vertices of $Q_{ij}$ in a Red-Black tree $T_{ij}$ with $w_i+1-y(v)$ as the key of each such vertex $v$. Moreover, for each $u\in R_0$, we keep an ordered set of indexes $I(u)=\{(i,j)\mid u\in P_{ij}\}$. Similarly, define the index set $I(v)$ for $v\in B_0$. We also store the vertices of $B_1$ in a Red-Black tree $T_1$ with $w_{\ell}+1-y(v)$ as the key of $v \in B_1$. 

\begin{observation}
MATCH uses $O(n \cdot poly(\log n))$ construction time and space.
\end{observation}

The above space bound follows from the fact that the space complexity of MATCH is dominated by the space needed for the Red-Black trees, which is $O(|B_1|)+\sum_{(i,j)} O(|Q_{ij}|)=O((n/{\epsilon}^2)\log^3 n)$, as a Red-Black tree uses linear space. The time bound follows trivially. 
\vspace{-4mm}

\subparagraph*{FIND$\_$MAXIMAL$\_$APS.} Let $F$ be the set of free vertices and $\Pi$ be the set of vertices that are already marked. Initially $\Pi=\emptyset$. Let $\mathcal{A}$ be the set of augmenting paths found so far, which is initialized to $\emptyset$. For each vertex $u\in R_1$, set its $E_1$-failed flag to 0. While there is a vertex $r_1\in (R\cap F)\setminus \Pi$, do the following.
\begin{itemize}
    \item Initialize a path $P$ at $r_1$. 
    \item While $P$ is not empty, do the following. 
    \begin{itemize}
        \item Let $P=\{r_1,b_1,\ldots,r_{\tau-1},b_{\tau-1},r_{\tau}\}$ be the current augmenting path that we need to extend. 
        \item (Case 1. $r_{\tau}\in R_0$) Access the array $A_2$ to find whether the copy of $r_{\tau}$ in $B_1$, i.e., $\hat{r}_{\tau}$, is marked and $(r_{\tau},\hat{r}_{\tau})$ is eligible. 
        \begin{itemize}
            \item If $\hat{r}_{\tau}$ is unmarked and $(r_{\tau},\hat{r}_{\tau})$ is eligible, call the subroutine AUGMENTING$\_$PATH ($\hat{r}_{\tau}$). If this subroutine returns DONE, terminate this while loop. Otherwise, jump to the next iteration. 
            \item Otherwise, search the Red-Black tree $T_{ij}$ where $(i,j)$ is the first index in $I(r_{\tau})$, to find a vertex $b_{\tau}$ with key value $y(r_{\tau})$. If such a vertex $b_{\tau}$ is found, call the subroutine AUGMENTING$\_$PATH($b_{\tau}$). Remove $b_{\tau}$ from all the Red-Black trees with indexes in $I(b_{\tau})$, as it is marked in the subroutine. If this subroutine returns DONE, terminate this while loop. Otherwise, jump to the next iteration. If no such vertex $b_{\tau}$ is found in $T_{ij}$, remove the index $(i,j)$ from $I(r_{\tau})$ and repeat the above step (performed for $(i,j)$) for the next index in  $I(r_{\tau})$. If $I(r_{\tau})$ becomes empty, remove $r_{\tau}$ and $b_{\tau-1}$ from $P$, and continue to the next iteration.
        \end{itemize}

        \item (Case 2. $r_{\tau}\in R_1$) Access the array $A_3$ to find whether the original copy of $r_{\tau}$ in $B_0$, say $v$, is marked and $(r_{\tau},v)$ is eligible. If $v$ is unmarked and $(r_{\tau},v)$ is eligible, call the subroutine AUGMENTING$\_$PATH($v$). If this subroutine returns DONE, terminate this while loop. Otherwise, jump to the next iteration. If $E_1$-failed is not set, i.e., it is 0 for $r_{\tau}$, search the Red-Black tree $T_1$ to find a vertex $b_{\tau}$ with key value $y(r_{\tau})$. If such a vertex $b_{\tau}$ is found, call the subroutine AUGMENTING$\_$PATH($b_{\tau}$). Remove $b_{\tau}$ from $T_1$, as it is marked in the subroutine. If this subroutine returns DONE, terminate this while loop. Otherwise, jump to the next iteration. If no such vertex $b_{\tau}$ is found in $T_1$, set $E_1$-failed flag for $r_{\tau}$ to 1, remove $r_{\tau}$ and $b_{\tau-1}$ from $P$, and continue to the next iteration. 
    \end{itemize}
\end{itemize}

The above procedure is self-explanatory. Next, we prove its correctness and bound the implementation time.  

\begin{restatable}{lemma}{matchcorrectness}\label{lem:match-correctness}
FIND$\_$MAXIMAL$\_$APS correctly computes a maximal set of disjoint augmenting paths. 
\end{restatable}

\begin{proof}
The proof of the lemma follows from the proof of correctness of the Gabow-Tarjan algorithm, assuming all eligible edges are considered if needed, while extending the path $P$ from $r_{\tau}$. Consider any eligible edge $e=(r_{\tau},v)$. Here we show that $e$ is considered for extension of $P$. If $v$ is marked, $v$ was already included at some augmenting path. As we are looking for disjoint paths, we do not need to consider the edge $(r_{\tau},v)$, and hence the correctness follows in this case. Thus, wlog, we can assume that $v$ is unmarked. First, assume $r_{\tau}\in R_0$. Then if $v\in B_1$, $e$ is explicitly considered by the subroutine. Otherwise, $v\in B_0$. In this case, suppose $(i,j)$ be an index such that $e\in (P_{ij},Q_{ij})$. Note that initially $(i,j)$ is in $I(r_{\tau})$. Thus, if needed, $T_{ij}$ can be searched, and as $v$ is unmarked, $v\in T_{ij}$. Now, as $e$ is eligible and unmatched (by Observation \ref{obs:unmatched_only}), $y(r_{\tau})+y(v)=w_{i}+1$. Thus, the key of $v$, $w_i+1-y(v)$ must be equal to $y(r_{\tau})$ and $v$ can be found eventually while searching $T_{ij}$. Now, assume $r_{\tau}\in R_1$. Then if $v\in B_0$, again $e$ is explicitly considered by the subroutine. Otherwise, $v\in B_1$. In this case, $v$ is in the tree $T_1$, as $v$ is unmarked. Now, as $e$ is eligible and unmatched (by Observation \ref{obs:unmatched_only}), $y(r_{\tau})+y(v)=w_{\ell}+1$. Thus, the key of $v$, $w_{\ell}+1-y(v)$ must be equal to $y(r_{\tau})$ and $v$ can be found eventually while searching $T_1$.   
\end{proof}

\begin{restatable}{lemma}{matchtime}\label{lem:match-time}
FIND$\_$MAXIMAL$\_$APS can be implemented in $O(n\cdot poly(\log n))$ time. 
\end{restatable}

\begin{proof}
For any vertex $u\in R$, let $L_u$ be the total time spent over the course of the procedure to extend paths that end at $u$, in particular, to find eligible edges of the form $(u,v')$ with $v'\in B$. Also, for $v\in B$, let $K_v$ be the total time spent to remove $v$ from the Red-Black tree(s) once it is marked. Note that the time of FIND$\_$MAXIMAL$\_$APS is dominated by the time $\sum_{u\in R} L_u+\sum_{v\in B} K_v$. Let us analyze $\sum_{u\in R} L_u$ at first. In the first case, $u\in R_0$. As mentioned before, probing of the edges in $E_2\cup E_3$ can be done in $O(n)$ time, using the arrays $A_2$ and $A_3$. So, let us focus on the remaining edges. Note that when we try to find an eligible edge for $u$, we search in $T_{ij}$, where $(i,j)\in I(u)$. Either no desired vertex $v$ is found, in which case, $(i,j)$ is removed from $I(u)$. Otherwise, a vertex $v\in Q_{ij}$ is found, but it is deleted from all Red-Black trees, and so it cannot appear in future search. Thus, removal of indexes from $I(u)$ and failed searches cost $O(| I(u)|\log n)$ time in total. The total number of vertices in $B_0$ that successfully appear in the result of the searches corresponding to all free vertices in $R_0$ is $\sum_{(i,j)} |Q_{ij}|$. Hence, \[\sum_{u\in R_0} L_u = \sum_{u\in R_0} O(| I(u)|\log n)+\sum_{(i,j)} O(|Q_{ij}|\log n)=\sum_{(i,j)} O((|P_{ij}|+|Q_{ij}|)\log n).\] 
The last inequality follows, as the sum of the sizes of the index sets in $\{I(u)\mid u\in R_0\}$ is bounded by the sum of the sizes of the sets in $\{P_{ij}\}$. 

Now, consider the case when $u\in R_1$. In this case, when we try to find an eligible edge $(u,v)$ for $u$ with $v\in B_1$, we search in $T_1$. Either no such $v$ is found, in which case, $E_1$-failed flag for $u$ is set to 1, and $T_1$ is never searched w.r.t. $u$ again. Or, such a vertex $v$ is found, in which case, it is deleted from $T_1$, and so it cannot appear in future search. Hence, \[\sum_{u\in R_1} L_u = O(|R_1|\log n)+O(|B_1|\log n)=O(n\log n).\] 

For $v\in B_1$, $K_v$ is bounded by $O(\log n)$. For $v\in B_0$, $K_v$ is bounded by the time needed to delete $v$ from all $T_{ij}$ in all of which it appears. Thus, \[\sum_{v\in B} K_v=O(|B_1|\log n)+\sum_{(i,j)} O(|Q_{ij}|\log n).\]
As $\sum_{(i,j)} |P_{ij}|+|Q_{ij}|=O((n/{\epsilon}^2)\log^3 n)$, it follows that, \[\sum_{u\in R} L_u+\sum_{v\in B} K_v=O((n/{\epsilon}^2)\log^4  n).\] Hence, the lemma follows. 
\end{proof}

\subsection{Hungarian Search Data Structure}
Recall that in Hungarian search, we need to maintain two sets $R'\subseteq R$ and $B'\subseteq B$. Initially, $R'=\emptyset$ and $B'=B$. In each iteration, a vertex is added to $R'$ and removed from $B'$. Additionally, each vertex $v$ has a weight $\sigma_v$. In every iteration, the goal is to maintain the bichromatic closest pair, which is the pair $(u,v)\in R'\times B'$ with the minimum $c(u,v)-\sigma_u-\sigma_v$ value. 

Consider the bi-clique cover $\mathcal{C}$ and in particular a pair $(P_{ij},Q_{ij})$ in $\mathcal{C}$. Then any edge in $P_{ij}\times Q_{ij}$ has the same weight $w_i$. Thus, the pair $(u,v)$ with the maximum $\sigma_u+\sigma_v$ value has the minimum $c(u,v)-\sigma_u-\sigma_v$ value in $P_{ij}\times Q_{ij}$. In other words, it is sufficient to keep track of a vertex $u\in P_{ij}$ with the maximum $\sigma_u$ value and a vertex $v\in Q_{ij}$ with the maximum $\sigma_v$ value. 

Now, we describe the Hungarian Search data structure. For each index $(i,j)$ in $\mathcal{C}$, we construct two max-heaps $H_{ij}^r$ and $H_{ij}^b$, for $P_{ij}\cap R'$ and $Q_{ij}\cap B'$, respectively. Initially, $H_{ij}^r$ is empty and $H_{ij}^b$ contains all the vertices in $Q_{ij}$. The key value of each vertex $v$ is its weight $\sigma_v$. Note that using $H_{ij}^r$ and $H_{ij}^b$, the pair $(u,v)\in (P_{ij}\cap R')\times (Q_{ij}\cap B')$ with the maximum $\sigma_u+\sigma_v$ value can be found in $O(1)$ time. Similarly, we construct two heaps $H_1^r$ and $H_1^b$ for the vertices in $R_1\cap R'$ and $B_1\cap B'$, respectively. Initially, $H_1^r$ is empty and $H_1^b$ contains all the vertices in $B_1$. The key value of each vertex $v$ is its weight $\sigma_v$. Again using $H_1^r$ and $H_1^b$, the pair $(u,v)\in (R_1\cap R')\times (B_1\cap B')$ with the maximum $\sigma_u+\sigma_v$ value can be found in $O(1)$ time. We also construct another max-heap $H_{23}$ to store the edges in $E_2\cup E_3$. Initially, it is empty. In every iteration, it contains only those edges $(u,v)$ such that $u\in R'$ and $v\in B'$. Finally, to keep track of the global maximum pair, we create another max-heap $H$. $H$ stores a maximum pair of $(P_{ij}\cap R')\times (Q_{ij}\cap B')$ for each index $(i,j)$ and maximum pairs of $(R_1\cap R')\times (B_1\cap B')$ and $H_{23}$.  

Note that the space complexity of the  data structure is $\sum_{(i,j)} O(|P_{ij}|+|Q_{ij}|)+O(n)=O((n/{\epsilon}^2)\log^3 n)$. This follows, as a max-heap uses linear space.

\begin{observation}
The Hungarian Search  data structure uses $O(n \cdot poly(\log n))$ construction time and space.
\end{observation}

Next, we describe a procedure which will help us implement Hungarian search. This procedure takes a pair $(u,v)$ as input, where $u\in R\setminus R'$ and $v\in B'$. We need to add $u$ to $R'$ and remove $v$ from $B'$ while maintaining the correct maximum pair. 

\subparagraph*{UPDATE$\_$CLOSEST$\_$PAIR$(u,v)$.} If $v\in B_0$, remove $v$ from all $H_{ij}^b$ that contains $v$. Also remove $(\hat{v},v)$ from $H_{23}$ if its in $H_{23}$. If $v\in B_1$, remove $v$ from $H_1^b$. Also remove the edge $(u',v)$ with $u'\in R_0$ from $H_{23}$ if its in $H_{23}$. If $u\in R_0$, add $u$ to all $H_{ij}^r$ such that $P_{ij}$ contains $u$. Also add the edge $(u,\hat{u})$ to $H_{23}$ if $\hat{u}\in B'$. If $u\in R_1$, add $u$ to $H_1^r$. Also add the edge $(u,v')$ with $v'\in B_0$ to $H_{23}$ if $v'\in B'$. Update the maximum pairs in $H$ accordingly by selecting the updated maximum pairs from the other max-heaps.  

\begin{restatable}{lemma}{hstime}\label{lem:hstime}
Hungarian search can be performed in $O(n \cdot poly(\log n))$ time. 
\end{restatable}

\begin{proof}
We note that there are at most $n$ iterations where in each iteration a vertex $u\in R\setminus R'$ is added to $R'$ and a vertex $v\in B'$ is removed from $B'$. Moreover, each vertex of $R$ is added to the data structure only once and each vertex of $B$ is removed from the data structure only once. To implement each iteration, we make a call to UPDATE$\_$CLOSEST$\_$PAIR($u,v$). The correctness follows trivially by construction of the data structure. Next, we analyze the time complexity. First, note that all the insertions and deletions in $H_1^b$, $H_1^r$ and $H_{23}$ can be performed in total time $O(n\log n)$. Also, the insertions of the vertices in $R$ to all $H_{ij}^r$ takes in total $\sum_{(i,j)} O(|P_{ij}|\log n)$ time. Similarly, the deletions of the vertices in $B$ from all $H_{ij}^b$ takes in total $\sum_{(i,j)} O(|Q_{ij}|\log n)$ time.  Note that an update to an element (maximum pair) in $H$ is triggered by an update in another max-heap. As we can have $O(n)+\sum_{(i,j)} O(|P_{ij}|+|Q_{ij}|)$ updates in total in other max-heaps, the updates in $H$ can be performed in total time $O(n\log n)+\sum_{(i,j)} O(|P_{ij}|+|Q_{ij}|)\log n$. Thus, the time taken over all iterations is $O(n\log n)+\sum_{(i,j)} O(|P_{ij}|+|Q_{ij}|)\log n=O((n/{\epsilon}^2)\log^4  n)$, which completes the proof of the lemma.         
\end{proof}

\bibliography{bib}
\end{document}